\documentclass[useAMS,usenatbib]{article}

\usepackage{natbib}
\usepackage{xcolor}
\usepackage{bbold}
\usepackage{amssymb}
\usepackage{amsmath}
\usepackage{graphicx}

\newtheorem{theorem}{Theorem}

\newcommand{\ew}{\mathbb{E}}
\newcommand{\qw}{\mathbb{P}}

\newcommand{\T}{\mathrm{\scriptscriptstyle T}}

\newcommand{\bo}{\boldsymbol{1}}
\newcommand{\bb}{\boldsymbol{\beta}}

\title{Robust Sequential Hypothesis Testing \\with Generalized Estimating Equations for \\Incomplete Clustered and Longitudinal Data}
\author{Nathan T. Provost$^*$ and
Abdus S. Wahed$^{**}$\\
\footnotesize Department of Biostatistics and Computational Biology,\\  \footnotesize School of Medicine and Dentistry, University of Rochester,\\
 \footnotesize Rochester, New York, U.S.A.\\
 \footnotesize$^*$nathan\_provost@urmc.rochester.edu \ $^{**}$abdus\_wahed@urmc.rochester.edu}
 \date{}

\begin{document}

\maketitle

\normalsize

\begin{abstract}
Existing sequential generalized estimating equation methodology for longitudinal and group-correlated data focuses on narrow hypotheses concerning treatment efficacy and often makes modeling assumptions that impede the desirable robustness of the involved test statistics. Drawing upon the well-established theory of incremental information gain for well-posed sequential analyses, we develop an approach that does not rely on modeling assumptions that infringe upon the robustness of the resulting estimators while simultaneously testing a much wider range of hypotheses. Our methodology provides general submatrix-level asymptotic theory for the evaluation of joint covariance matrices of sequential test statistics. Moreover, this framework allows us to construct a novel approach to computing efficacy boundaries, the likes of which can be estimated with greater precision at later interim times. These constructions also accommodate accessible multiple imputation procedures, thereby allowing for our approach to be applied to incomplete datasets. Type I error and power are assessed through a series of comprehensive simulations mirroring the simulations of recent work to facilitate a proper comparison. We conclude by applying our methods to a dataset from a longitudinal study concerning the impact of race on the efficacy a treatment for hepatitis C. 
\end{abstract}

\section{Introduction}
\label{s1}

Many prospective biomedical research studies, including clinical trials, involve repeated measures of the same outcome over time \cite[e.g.,][]{hepc2,hepc}. Although these studies are planned with fixed recruitment and follow-up periods, a key goal is to reach a valid conclusion as early as possible using the accumulated data. Doing so conserves resources and avoids exposing patients to ineffective treatments. The formal approach to evaluating the hypothesis at planned points during data collection is known as interim monitoring or group-sequential analysis.

Extensive previous work addressed the sequential analysis for longitudinal or clustered outcomes (e.g., \citet{asw,led,led95,lkt}). Compared to the case of univariate outcomes, a somewhat greater challenge with group-sequential analysis of such outcomes is accounting for the correlation among repeated measures. Early intuitive approaches based on accumulation statistics, such as those proposed by \citet{asw} and \citet{dng}, rely on the concept of increasing information over time, but suffer from intensive linearity and outcome normality assumptions \citep{jtg}, limiting their applications in modern analyses. Later, \citet{led} and \citet{led95} used mixed-effects models to broaden these methods, and the concept of information ratios, often simplified to ratios of interim group sizes, was further developed in works such as \citet{led95} and \citet{lkt}.

\citet{wsl} introduced group-sequential analysis for longitudinal data using the GEE framework of \citet{lz}, initially assuming an independent group correlation structure. This was later generalized by \citet{lkt} to allow broader correlation structures, though their theory required correct specification of the group correlation. The notions of accumulating information and independent increments, discussed in \citet{lkt} and foreshadowed in \citet{led95}, were formalized by \citet{jt}, who showed that the covariance between two score statistics at different times equals the variance at the earlier time. These developments established the foundation for group-sequential analysis with longitudinal data, with formal efficiency results presented in \citet{str}.

The hypotheses and model settings in these works were somewhat restrictive. \citet{lkt} considered only a two-sided hypothesis for a single scalar treatment effect, treating all other effects as nuisance parameters, allowing them to use semiparametric profile scores to induce a special matrix decomposition that underlies their methods and subsequent extensions using the same framework \cite[e.g.,][]{lbx}. \citet{slvk, slv} summarized these results, highlighting the cleanly partitioned matrix structure arising from the independent increments phenomenon described in earlier work \cite[e.g.,][]{led95}. \citet{lgb} used quadratic-form theory to study the scalar variances and covariance matrices of summary statistics formed by concatenating sequential $\chi^2$ statistics. \citet{jlw} improved family-wise error control for multiple testing, while \citet{cl} extended \citeauthor{led}’s linear mixed-model framework \citep{led} to a more general setting. Other contemporary work \cite[e.g.,][]{kse} used response-curve summaries based on linear mixed-effects models. A sequential sampling approach for longitudinal analysis was proposed by \citet{pc}. \citet{se} and \citet{sre} examined the performance of sequential methods derived from \citet{lz} when the independent-increments structure fails (e.g., due to efficiency violations; \citet{str}).

Overall, prior methods have largely focused on narrow hypotheses about treatment effects, treating other covariates as nuisance parameters. As a result, studies targeting more nuanced hypotheses are not well served by this framework. For example, investigators may be interested in higher-order interactions involving nuisance parameters and target estimands, such as treatment-by-time effects within predefined subgroups. \citet{jtg} considered a model including a treatment-by-time interaction adjusted for baseline covariates. Although this extends existing group-sequential methods, the hypothesis remains somewhat specific, and the maximum test statistic applies only when interest is limited to detecting any treatment effect at pre-specified, discrete time points. Moreover, the lack of a closed-form distribution for this maximum-based statistic limits further theoretical development.

Our work addresses these challenges by leveraging concepts closely related to the information fractions in \citet{led95}, \citet{lkt}, and \citet{slv}. Using a compound estimating equation based on the constructions of \citet{lkt}, we show that the full covariance matrix of the sequential test statistics can be estimated at any interim time using basic components of the sandwich estimator of \citet{lz}. This estimated covariance matrix enables computation of efficacy boundaries similar to those of \citet{p77} and \citet{obf}, after which interim monitoring inference proceeds in a standard manner. Unlike classical approaches such as \citet{lkt}, our method does not require a correctly specified working correlation matrix, thereby retaining the robustness of the original GEE framework of \citet{lz}. Although \citet{td} also achieves robustness through linear combinations of potentially inefficient estimators, pontificating that the efficiency condition of \citet{str} is sufficient but not necessary for independent increments, their work focuses on simple treatment-effect hypotheses. In contrast, our approach allows more flexible modeling and hypothesis testing, improves upon the restrictive systematic component in \citet{jtg}, and yields test statistics with closed-form $\chi^2$ distributions at each analysis time.

We also address missing data, an issue of unique difficulty in sequential analyses. Historical GEE-based methods typically assume observations are missing completely at random (MCAR) as in \citet{lz}. \citet{jtg} adopts this assumption, though a weaker missing-at-random (MAR) condition can be used under correct covariance specification. Our method can be readily combined with multiple-imputation procedures for GEE, such as the approach of \citet{r87}, providing a more general treatment of missing data than existing work like \citet{jtg} in that the working correlation need not be correctly specified when accounting for MAR data in our methods.

In this article, Section \ref{s2} reviews background on GEE and sequential testing. Section \ref{s3} presents our main methodological contributions, focusing on the joint distribution of the proposed test statistics. Section \ref{s4} reports simulation results, Section \ref{s5} presents the VIRAHEP-C study analysis, and Section \ref{s6} concludes with a discussion.

\section{Classical GEE and Group Sequential Tests}
\label{s2}
\subsection{Classical GEE}
\label{ss21}

Classical generalized estimation equation (GEE) models, as initially described by \citet{lz}, account for within-group correlation structures when analyzing the relationship between outcomes and a set of covariates for a prespecified link function. Specifically, the traditional GEE approach examines $n$ groups (indexed by $i=1,\dots,n$), the $i$th of which has $s_i$ total observations and an outcome vector $Y_i$ with $s_i$ components. Let $X_{ik}$, with $\mathrm{dim}(X_{ik})=p$ be the covariate vector corresponding to the $k$th observation in group $i$ and $X_i=[X_{i1}^\T \ \hdots \ X_{is_i}^\T]$ be the group specific covariate matrix. For some chosen link function $\mu$, GEE posits that $\ew(Y_{ik}\mid X_{ik})=\mu(X_{ik},\beta)$, where $Y_{ik}$ is the $k$th component of $Y_i$ and $\beta$ is a $p$-component vector of estimable coefficients. \citet{lz} showed that the solution $\hat{\beta}$ to the equation
\begin{equation}\sum_{i=1}^nD_{i}^\T V_i^{-1}(Y_i-\mu_{i})=0,\label{gee}\end{equation}
where $\mu_i\equiv(\mu(X_{i1},\beta),\hdots,\mu(X_{is_i},\beta))^\T$, $V_i\equiv V_i(X_i,\beta)$ is a group or individual specific $s_i\times s_i$ covariance matrix potentially dependent on covariate $X_i$ and parameter $\beta$, and $D_i\equiv D(X_{i},\beta)=\dot{\mu}(X_i,\beta)$ denotes the Jacobian matrix of $\mu(X_i,\beta)$ with respect to $\beta$, satisfies $\sqrt{n}(\hat{\beta}-\beta)\overset{d}{\rightarrow}\mathrm{MVN}(0,\Sigma)$, where
\small
\[\Sigma=\lim_{n\rightarrow\infty}n\left(\sum_{i=1}^nD_i^\T V_i^{-1}D_i\right)^{-1}\left\{\sum_{i=1}^nD_{i}^\T V_i^{-1}\mathrm{Cov}(Y_i) V_i^{-1} D_i\right\}\left(\sum_{i=1}^nD_i^\T V_i^{-1}D_i\right)^{-\T},\] 
\normalsize and $\mathrm{Cov}(Y_i)$ denotes the true outcome covariance matrix for group $i$. Moreover, they showed that $\Sigma$ can be consistently and robustly estimated by $\hat{\Sigma}=\hat{\Omega}^{-1}\hat{\Lambda}\hat{\Omega}^{-\T}$, where
\begin{equation}\hat{\Omega}=-\frac{1}{n}\sum_{i=1}^n\hat{D}_i^\T \hat{V}_i^{-1}\hat{D}_i,\quad\hat{\Lambda}=\frac{1}{n}\sum_{i=1}^n\hat{D}_{i}^\T \hat{V}_i^{-1}(Y_i-\hat{\mu}_{i})(Y_i-\hat{\mu}_{i})^\T \hat{V}_i^{-1} \hat{D}_i,\label{omla}\end{equation} 
where we have used a hat $(~\hat{}~)$ on functionals to indicate that the unknown parameters in those functionals are replaced with the corresponding consistent estimates, (e.g., $\hat D_i\equiv D(X_{i},\hat{\beta})$). We note also that $\hat{\Omega}^{-1}$ also provides a (na\"ive) estimate of the variance, but this is subject to bias under an incorrectly specified $V_i$. Applying the multivariate delta method, it follows that for any invertible function $h:\mathbb{R}^p\rightarrow\mathbb{R}^q$ with a well-posed $q\times p$ Jacobian matrix at $\beta$ given by $\dot{h}(\beta)$, the test statistic $T=n[\{h(\hat{\beta})-\gamma\}^\T\{\dot{h}(\hat\beta)\hat\Sigma\dot{h}^\T(\hat\beta)\}^{-1}\{h(\hat{\beta})-\gamma\}]$ can be used to test $H_0: h(\beta)=\gamma$ versus $H_a: h(\beta)\neq \gamma$ for any $\gamma\in\mathbb{R}^q$, since it holds that $T\sim\chi^2[\mathrm{rank}\{\dot{h}(\hat{\beta})\hat{\Sigma}\dot{h}^\T(\hat{\beta})\}]$ under $H_0$. For most clinical applications, it would suffice to choose $\gamma=\mathbb{0}_q$ and $h(\beta)=A\beta$ for some $q\times p$ contrast matrix $A$, resulting in the simplified test statistic $T=n(A\hat\beta)^\T(A\hat\Sigma A^\T)^{-1}A\hat{\beta}.$

\subsection{Na\"ive Group-Sequential Tests}
\label{ss22}
Suppose a maximum of $M$ analyses are planned throughout the study to test the hypothesis $H_0$ based on data accumulated up to the analysis time $t_m, m = 1, 2, \ldots, M$ $(t_i < t_j, 1\leq i < j\leq M)$. Thus, there are $M-1$ planned interim analyses with the $M$th analysis being the full study setting. For each analysis time $m$, suppose that we only have $n_m\leq n_M=n$ of the groups at our disposal to estimate the target estimand $\beta=(\beta_1,\dots,\beta_p)^\T$. For this development, we assume that all $n_m$ groups have complete data. The case with missing or partial information will be discussed later. Denote the estimate of $\beta_j$ constructed with all information available at time $t_m$ by $\hat{\beta}_{mj}$ and let $\hat{\beta}_m=(\hat{\beta}_{m1},\hdots,\hat{\beta}_{mp})^\T$. Hence, if each $\hat{\beta}_m$ were obtained by constructing a GEE model from \citet{lz} with the currently available $n_m$ groups, it would hold that since $\sqrt{n_m}\{h(\hat{\beta}_m)-\gamma\}\overset{d}{\rightarrow}\mathrm{MVN}\{0,\dot{h}(\beta)\Sigma \dot{h}^\T(\beta)\}$ if $h(\beta)=\gamma$, the test statistics
\begin{equation}T_m=n_m[\{h(\hat{\beta}_m)-\gamma\}^\T\{\dot{h}(\hat{\beta}_m)\hat\Sigma_m\dot{h}^\T(\hat{\beta}_m)\}^{-1}\{h(\hat{\beta}_m)-\gamma\}]\quad (m=1,\dots,M)\label{eq:falset}\end{equation}
could be used to test $H_0: h(\beta)=\gamma$ against $H_a: h(\beta)\neq \gamma$, where $\hat{\Sigma}_m$ is defined analogously to the estimator from \eqref{omla} with $n_m$ replacing $n$ and $\hat{\beta}_m$ replacing $\hat\beta$. A na\"ive approach would be to use each test statistic at its respective interim and perform a $\chi^2$ test with $\mathrm{rank}\{\dot{h}(\hat{\beta}_m)\hat\Sigma_m\dot{h}^\T(\hat{\beta}_m)\}$ degrees of freedom. As shown in historical literature \cite[e.g.,][]{a69}, this inflates the type I error of the test. The majority of sequential analysis methods, including historical and contemporary works, focus on abating this issue by computing the joint distribution of these test statistics and using it to compute adjusted critical values in accordance with a predetermined efficacy boundary approach \citep{p77,obf}. In the next section, we develop our method for obtaining this joint distribution and computing such boundaries.

\section{Joint Distribution of Sequential Statistics \\ Through Compound GEE}
\label{s3}
\subsection{Joint Distribution of Sequential Test Statistics}
\label{ss31}
 
Suppose that the subjects are indexed according to their arrival in the study, which is presumed to be random in order. Let $\xi_{im}$ be an indicator of whether or not subject $i$ arrives by time $t_m$, $\xi_i = (\xi_{i1}, \ldots, \xi_{iM})^\T$, and denote $\psi(X_i; {\beta})=D^\T(X_i,\beta)V^{-1}(X_i,\beta)(Y_i-\mu_{i}).$ Then the $M$ estimators of $\beta$ from the $M$ analyses, stacked as the concatenated $Mp \times 1$ vector
$\hat{\boldsymbol{\beta}}=(\hat{\beta}_1^\T,\dots,\hat{\beta
}_M^\T)^\T$, can be written as a solution for the argument $\boldsymbol{\beta} = ({\beta}_1^\T,\dots,{\beta}_M^\T)^\T$ satisfying the compound estimating equation 
\begin{equation}\sum_{i=1}^n\Psi_i (X_i, \xi_{i} ; \boldsymbol{\beta})=0, \label{psi}\end{equation}
where $
\Psi_i (X_i,\xi_{i} ; \boldsymbol{\beta}) = (\xi_{i1}\psi^\T_{i1}, \ldots, \xi_{iM}\psi^\T_{iM})^\T$ with $\psi_{im}\equiv \psi(X_i,{\beta_m}),$ for every $ m = 1, 2, \ldots, M$, is the stacked kernel of estimating equations from $M$ analyses. We note that $\xi_{iM}=1$ for all $i=1, \ldots, n$, by definition, and that each subvector $\xi_{im}\psi_{im}$ was originally considered separately in \citet{wsl} and \citet{lkt}.  Each $\hat{\beta}_m$ estimates the same estimand $\beta$, so $\hat{\boldsymbol{\beta}}$ estimates $\bo_M\otimes\beta$, where $\bo_M=\mathrm{diag}(\boldsymbol {I}_M)$, $\boldsymbol {I}_M$ being an $M \times M$ identity matrix. 

Since \eqref{psi} is a valid estimating equation, following the theory laid out in \citet{lz} and the asymptotic properties of M-estimators from \cite{sb}, we have $\sqrt{n}\{\hat{\bb} - (\bo_M\otimes\beta)\}\overset{d}{\rightarrow}\mathrm{MVN}(0,\boldsymbol{\Sigma})$ with the $Mp\times Mp$ covariance matrix 
\begin{equation}
\boldsymbol{\Sigma} = \boldsymbol\Omega^{-1}\boldsymbol\Lambda\boldsymbol\Omega^{-\T},\label{eq:boldsigma}\end{equation}  
where
\begin{equation}
\boldsymbol{\Omega}=-\lim_{n\rightarrow\infty}\frac{1}{n}\sum_{i=1}^n \ew\left[\dot{\Psi}_i\{X_i,\xi_i;(\bo_M\otimes\beta)\} \ \Big| \ X_i\right],
\label{eq: boldomega}\end{equation}
with $\dot{\Psi}_i\{X_i,\xi_i;(\bo_M\otimes\beta)\}$ being the Jacobian of $\Psi_i (X_i,\xi_{i} ; \bb)$ evaluated at $\bo_M\otimes\beta$, and
\begin{equation}
\boldsymbol{\Lambda}=\lim_{n\rightarrow\infty}\frac{1}{n}\sum_{i=1}^n \ew \left[\Psi_i \{X_i,\xi_{i} ; (\bo_M\otimes\beta)\} \Psi_i^\T \{X_i,\xi_{i} ;(\bo_M\otimes\beta)\}\mid X_i\right].
\label{eq: boldlambda}\end{equation}
These constructions are analogous to the components outlined in \eqref{omla}. Therefore, the asymptotic joint distribution of $\hat{\boldsymbol{\beta}}$ can be fully characterized through the knowledge of true $\beta$ and $\boldsymbol{\Sigma}$. Had the full sample data been available, $\boldsymbol{\Sigma}$ could have been estimated by the usual robust sandwich estimator defined using \eqref{omla}.  At any analysis time $m<M$, though, all information, specifically the data from the $(n_m+1)$th to the $n$th subject, will be unavailable and hence such an estimator cannot be directly calculated based on data at the $m$th interim. To address this challenge, we take advantage of the structural separation among the stacked elements of the estimating equation \eqref{psi}. More specifically, the $m$th $p \times 1$ block of $
\Psi_i (X_i,\xi_{i} ; \boldsymbol{\beta})$, namely, $\xi_{im}\psi_{im}$ involves only $\beta_m$. This ensures that the matrix $\boldsymbol{\Omega}$ can be written as a block-diagonal matrix  as $\boldsymbol{\Omega} = \mathrm{diag}(\Omega_1, \ldots, \Omega_M)$, where, since $\ew(\xi_{im}\mid X_i)=\xi_{im}$,
\small
\begin{equation}\Omega_{m}=-\lim_{n\rightarrow\infty}\frac{1}{n}\sum_{i=1}^n \ew\left\{\xi_{im}\dot{\psi}(X_i;\beta_m) \ \Big| \ X_i\right\}=-\lim_{n\rightarrow\infty}\frac{1}{n}\sum_{i=1}^n \xi_{im}\ew\left\{\dot{\psi}(X_i;\beta_m) \ \Big| \ X_i\right\},\label{eq:omm}\end{equation}
\normalsize
for all $m=1,\dots, M$, with $\dot{\psi}(X_i;\beta_m)$ being the Jacobian of $\psi(X_i;\beta_m)$ with respect to $\beta_m$. Second, we note that $\xi_{im'}=1$ implies $\xi_{im} =1$ for all $m^\prime < m,$ and $m, m^\prime \in\{ 1, 2, \ldots, M\},$ which allows us to write $\boldsymbol{\Lambda}$ as a $Mp\times Mp$ block matrix given by $\boldsymbol{\Lambda} = (\Lambda_{mm^\prime})$, where $\Lambda_{mm^\prime}$ is given by
\begin{equation}\lim_{n\rightarrow\infty}\frac{1}{n}\sum_{i=1}^n \xi_{im'} D^\T(X_i,\beta_m) V^{-1}(X_i,\beta_m)\mathrm{Cov}(Y_i)V^{-1}(X_i,\beta_{m'})D(X_i,\beta_{m'}),\label{eq:lmm}\end{equation}
for any $m^\prime \leq m$ with $\Lambda_{mm^\prime} = \Lambda^\T_{m^\prime m}$, both being $p\times p$ submatrices. This natural decomposition of $\boldsymbol{\Omega}$ and $\boldsymbol{\Lambda}$ enables us to write $\boldsymbol{\Sigma}$ as a block matrix
$\boldsymbol{\Sigma} = (\Sigma_{mm^\prime})$, where each submatrix has the form $\Sigma_{mm^\prime} = \Omega_m^{-1}\Lambda_{mm^\prime}\Omega_{m^\prime}^{-\T}, \;m, m^\prime \in\{ 1, 2, \ldots, M\}.$ This block representation aids in the approximation of the entire $Mp\times Mp$ covariance matrix as we show later, since each submatrix mirrors the form of the robust covariance matrix defined in \eqref{omla}.

From these constructions, the asymptotic joint distribution of the standardized GEE coefficient estimates at the $m$th interim can be characterized using the asymptotic distribution of the stacked standardized coefficient vector ($\hat{\bb}$). Let $\eta_m=\sqrt{n_m}(\hat{\beta}_m-\beta)$, $\Pi_m=(\sqrt{n_m/n})\boldsymbol{I}_p$ for all $m=1,\dots, M$, $\boldsymbol{\eta}=(\eta_1,\dots,\eta_M)^\T$, 
$\boldsymbol{\omega}=\sqrt{n}\{\hat{\bb} - (\bo_M\otimes\beta)\}$, and $\boldsymbol{\Pi}=\mathrm{diag}(\Pi_1,\dots,\Pi_M)$. Then rudimentary matrix algebra guarantees that $\boldsymbol{\eta}=\boldsymbol{\Pi}\boldsymbol{\omega}$, and  since $\boldsymbol{\omega}\overset{d}{\rightarrow}\mathrm{MVN}(0,\boldsymbol{\Sigma})$, 
it holds by Slutsky's theorem that $\boldsymbol{\eta}\overset{d}{\rightarrow}\mathrm{MVN}(0,\boldsymbol{\Gamma})$ where $\boldsymbol{\Gamma}=\boldsymbol{\Pi}\boldsymbol{\Sigma}\boldsymbol{\Pi}$. Since $\boldsymbol{\Pi}$ is a block diagonal matrix, $\boldsymbol{\Gamma}$ inherits the useful block decomposition of $\boldsymbol{\Sigma}$, that is, we can write $\boldsymbol{\Gamma}=(\Gamma_{mm'})$, where $\Gamma_{mm'}=\{(\sqrt{n_mn_{m'}})/n\}\Sigma_{mm'}$ for all $m,m'\in\{1,\dots,M\}$.

\subsection{Estimation Procedure for $\boldsymbol{\Sigma}$ Using Relative Information}
\label{ss32}

Using the compound score equation given in \eqref{psi} and the block matrix estimand expressions defined in \eqref{eq:omm} - \eqref{eq:lmm}, we can naturally define the estimators
\[\hat{\Omega}_m=-\frac{1}{n}\sum_{i=1}^n\xi_{im}D^\T(X_i,\hat{\beta}_m)V^{-1}(X_i,\hat{\beta}_m)D(X_i,\hat{\beta}_m)\]
\[\hat\Lambda_{kr}=\frac{1}{n}\sum_{i=1}^n\xi_{ir}D^\T(X_i,\hat{\beta}_k)V^{-1}(X_i,\hat{\beta}_k)(Y_i-\hat{\mu}_{ik})(Y_i-\hat{\mu}_{ir})^\T V^{-1}(X_i,\hat{\beta}_r)D(X_i,\hat{\beta}_r)\]
and $\hat{\Lambda}_{rk}=\hat{\Lambda}_{kr}^\T$ for $k\geq r$, where $\hat{\mu}_{im}=(\mu(X_{i1},\hat{\beta}_m),\dots,\mu(X_{is_i},\hat{\beta}_m))^\T$. These are to be contrasted with their standard GEE counterparts given by
\small
\[\hat{\Omega}_m^*=-\frac{1}{n_m}\sum_{i=1}^{n_m}D^\T(X_i,\hat{\beta}_m)V^{-1}(X_i,\hat{\beta}_m)D(X_i,\hat{\beta}_m)\]
\[\hat\Lambda_{m}^*=\frac{1}{n_m}\sum_{i=1}^{n_m}D^\T(X_i,\hat{\beta}_m)V^{-1}(X_i,\hat{\beta}_m)(Y_i-\hat{\mu}_{im})(Y_i-\hat{\mu}_{im})^\T V^{-1}(X_i,\hat{\beta}_m)D(X_i,\hat{\beta}_m),\]
\normalsize
both of which could be used to construct a robust variance estimator $(\hat{\Sigma}_m)$ in the traditional sense. The main result, stated below, utilizes the relationship between the standard GEE submatrices and our compound GEE submatrices, along with the ideas relating the two respective distributions as presented at the end of Subsection \ref{ss31}, to provide a general framework for estimating the overall joint covariance matrix of the test statistics by working from the submatrix scale upward.

\begin{theorem}
    Let $\pi_m=\ew(\xi_{im})$, $\Tilde{\Omega}_m=\pi_m\hat{\Omega}_m^*$ and $\Tilde{\Lambda}_{mm}=\pi_m\hat{\Lambda}_m^*$. Then $\Tilde{\Omega}_m$ and $\hat{\Omega}_m$ converge to the same limit in probability, $\Tilde{\Lambda}_{mm}$ and $\hat{\Lambda}_{mm}$ converge to the same limit in probability,  and, if the maximum number of groups examined at the final analysis is fixed, then for $m=1,\dots,M$ and $ 1\leq r\leq k\leq M$,
    \[\hat{\Omega}_r^{(m)}=\frac{n_r}{n_m}\Tilde{\Omega}_m,\ \ \hat{\Lambda}_{kr}^{(m)}=\frac{n_r}{n_m}\Tilde{\Lambda}_{mm},\ \ \hat{\Lambda}_{rk}^{(m)}=\left\{\hat{\Lambda}_{kr}^{(m)}\right\}^\T\]
    are consistent estimators for $\Omega_r \ (r=1,\dots,M)$, $\Lambda_{kr}$, and $\Lambda_{rk} \ (1\leq r\leq k\leq M)$, respectively. Hence,  $\hat{\boldsymbol{\Omega}}^{(m)}=\mathrm{diag}(\hat\Omega_1^{(m)},\dots,\hat\Omega_M^{(m)})$ and $\hat{\boldsymbol{\Lambda}}^{(m)}=\left[\hat \Lambda_{kr}^{(m)}\right]_{k,r=1,\dots,M}$ are both consistent estimators for $\boldsymbol{\Omega}$ and $\boldsymbol{\Lambda}$, and
    \[\hat{\boldsymbol{\Sigma}}^{(m)}=\left\{\hat{\boldsymbol{\Omega}}^{(m)}\right\}^{-1}\hat{\boldsymbol{\Lambda}}^{(m)}\left\{\hat{\boldsymbol{\Omega}}^{(m)}\right\}^{-\T}\]
    is a consistent estimator for $\boldsymbol\Sigma$ at any time $t_m$. $\quad\Box$
    
\end{theorem}

    \noindent\textbf{Proof}: Let $\mathcal{F}(X_i,\hat{\beta}_m)=\ew\{\xi_{im}\dot{\psi}(X_i,\hat{\beta}_m)\mid X_i\}$ and assume $\ew\left\{\dot{\psi}(X_i,\beta)\right\}=\mathcal{E}$, $|\mathcal{E}|<\infty$. Then $\ew\{\mathcal{F}(X_i,\hat{\beta}_m)\}=\ew\{\ew(\xi_{im}\mid X_i)\}\ew[\ew\{\dot{\psi}(X_i,\beta)\mid X_i\}]=\ew(\xi_{im})\ew\{\dot{\psi}(X_i,\beta)\}=\pi_m\mathcal{E}$. From \citet{lz}, each $\hat{\beta}_{mj}\overset{p}{\rightarrow}\beta_j$ and hence admits an approximation by averages in the sense outlined in Theorem 5.28 of \citet{bs}. Assume that each entry of $\mathcal{F}$ admits a Hessian matrix with respect to its parameter argument that in turn has entries whose absolute value is finite in expectation within a neighborhood of $\beta$, and that $\mathcal{F}$ has entries that are finite in variance and have gradients that are finite in expectation. Then by Theorem 5.28 of \citet{bs}, 
    \[\hat{\Omega}_m=-\frac{1}{n}\sum_{i=1}^n\xi_{im}D^\T(X_i,\hat{\beta}_m)V^{-1}(X_i,\hat{\beta}_m)D(X_i,\hat{\beta}_m)=-\pi_m\mathcal{E}+o_p(n^{-1/2}).\]
    By the same argument, $\hat{\Omega}_m^*=-\mathcal{E}+o_p(n_m^{-1/2})$, so $\pi_m\hat{\Omega}_m^*=-\pi_m\mathcal{E}+o_p(n_m^{-1/2})$. Thus, $\hat{\Omega}_m$ and $\Tilde{\Omega}_m=\pi_m\hat{\Omega}_m^*$ converge to the same limit in probability. If $n_m/n<\infty$ and is fixed, then $\pi_m=\ew[\xi_{im}]=n_m/n$ and $\pi_r/\pi_m=n_r/n_m$, so $(n_r/n_m)\Tilde{\Omega}_m=-\pi_r\mathcal{E}+o_p(n_m^{-1/2})$, so $(n_r/n_m)\Tilde{\Omega}_m$ converges to the same limit as $\hat{\Omega}_r$ in probability. Thus, $\hat{\boldsymbol{\Omega}}^{(m)}\overset{p}{\rightarrow}\boldsymbol{\Omega}$. This proof applies identically to $\hat{\boldsymbol{\Lambda}}^{(m)}$, hence it is clear that $\hat{\boldsymbol{\Lambda}}^{(m)}\overset{p}{\rightarrow}\boldsymbol{\Lambda}$, and thus, $\hat{\boldsymbol{\Sigma}}^{(m)} \overset{p}{\rightarrow}\boldsymbol{\Sigma}.$ $\hfill\Box$ 

Each $\sqrt{n}_r\hat{\beta}_r$ ($r=1,\dots,M$) has its own robust variance matrix estimate $\hat{\Sigma}_r$ obtained through standard GEE, which is equivalent to an estimate of $\Gamma_{rr}$ as discussed before Subsection \ref{ss32}. 

Recalling the discussion at the end of Subsection \ref{ss31}, since the marginal distributions must equate,
one can approximate $\hat{\Sigma}^{(m)}_{rr}$ with $(n/n_r)\hat{\Sigma}_r$ and, equivalently, $\hat{\Sigma}_r$ with $(n_r/n)\hat{\Sigma}^{(m)}_{rr}$, implying  $\hat{\Sigma}^{(m)}_{rr}/n=\hat{\Sigma}_r/n_r$. This dynamic is important for the discussion of Monte Carlo methods in Subsection \ref{ss33}. This is an informal justification of the prior result, rooted implicitly in the historical notion of information fractions as mentioned in \citet{lkt}, \citet{ld}, and \citet{slvk} among others, which in our notation is passively given by $\pi_m$.

\subsection{Boundary Calculation for Sequential Testing}
\label{ss33}

At the $m$th interim time, we can obtain an overall variance matrix estimator using the methods detailed in Theorem 1. We also have $\hat{\beta}_m=(\hat{\beta}_{m1},\dots,\hat{\beta}_{mp})^\T$, with which we define $\hat{\boldsymbol{\beta}}^{(m)}=\bo_M\otimes\hat{\beta}_m$. Recalling that every $\hat{\beta}_m$ is a consistent estimator for $\beta$, it follows that $\hat{\boldsymbol{\beta}}^{(m)}$ is a consistent estimator for $\bo_M\otimes\beta$. Therefore, it is now possible to empirically estimate the joint distribution of $T$ using $\hat{\boldsymbol{\beta}}^{(m)}$ and $\hat{\boldsymbol{\Sigma}}^{(m)}$ using Monte Carlo methods. As discussed in \citet{wt}, an efficacy boundary constitutes a set of critical values $\{\tau_m\}_{m=1}^M$ that satisfy $\mathcal{B}\left(\{\tau_m\}_{m=1}^M\right)=\alpha$, where the boundary function $\mathcal{B}\left(\{\tau_m\}_{m=1}^M\right)$ is given by
\begin{equation}\qw_{H_0}(T_1>\tau_1)+\sum_{m=2}^M\qw_{H_0}\left(T_j>\tau_j \ \bigg| \ \bigcap_{c=1}^{j-1}\{T_c\leq\tau_c\}\right)\qw_{H_0}\left(\bigcap_{c=1}^{j-1}\{T_c\leq\tau_c\}\right),\label{eq:bf}\end{equation}
where $\alpha$ is the specified type I error and $\qw_{H_0}$ is the probability measure under the null hypothesis. Since we have all the necessary parameters needed for the distribution of $\bb$, it is possible to draw from its approximated distribution and use these draws to estimate the value of $\mathcal{B}\left(\{\tau_m\}_{m=1}^M\right)$ for any specified $\tau_1,\dots,\tau_M$. In the case of the efficacy boundary set forth by \citet{p77}, which takes $\tau_1=\dots=\tau_M=\tau$, it is possible to find
\[\tau^*=\underset{\tau}{\mathrm{argmin}}\left\{|\mathcal{B}(\tau)-\alpha|\right\},\]
using traditional optimization methods if we have empirical estimates for all necessary probabilities. We note that under $H_0$, $h(\beta)=\gamma$, which implies that $\beta=h^{-1}(\gamma)$, and since $\hat{\beta}_m\overset{\qw}{\rightarrow}\beta$ for any $m$, $\hat{\beta}_m\overset{\qw}{\rightarrow} h^{-1}(\gamma)$. Noting that $h(\beta)-\gamma=0$ under $H_0$ and drawing some large number $B$ of samples $\{\beta^{(b)}\}_{b=1}^B$, where $\beta^{(b)}\sim\mathrm{MVN}\{\mathbb{0}_{Mp},\hat{\boldsymbol{\Sigma}}^{(m)}/n\}$ from which each $\beta^{(b)}$ can be written as $\beta^{(b)}=(\beta^{(b)}_1,\dots,\beta^{(b)}_M)^\T$ where each $\beta^{(b)}_r$ is a $p$-dimensional subvector, we can compute the $B$ total empirical test statistic draws $\{T^{(b)}\}_{b=1}^B=\{(T^{(b)}_1,\dots,T^{(b)}_M)^\T\}_{b=1}^B$, where 
\begin{equation}T_r^{(b)}=n[\{h(\beta_r^{(b)})-\gamma\}^\T\{\dot{h}(\beta_r^{(b)})\hat\Sigma_{rr}^{(m)}\dot{h}^\T(\beta_r^{(b)})\}^{-1}\{h(\beta_r^{(b)})-\gamma\}]\quad (r=1,\dots,M).\label{eq:mcts}\end{equation} 
Note that we have used the fact that $\hat{\Sigma}_{rr}^{(m)}/n=\hat{\Sigma}_r/n_r$ as discussed at the end of Subsection \ref{ss32} to expedite the computation of our Monte Carlo test statistics. This relationship guarantees that our representation given in \eqref{eq:mcts} is equivalent to the test statistics given in \eqref{eq:falset}.

With these draws, the empirically estimated \citet{p77} boundary function is
\begin{equation}\widehat{\mathcal{B}}_{\mathrm{P}}(\tau)=\widehat{\qw_{H_0}}(T_1>\tau)+\sum_{m=2}^M\widehat{\qw_{H_0}}\left(T_j>\tau \ \bigg| \ \bigcap_{c=1}^{j-1}\{T_c\leq\tau\}\right)\widehat{\qw_{H_0}}\left(\bigcap_{c=1}^{j-1}\{T_c\leq\tau\}\right),\label{eq:bfef}\end{equation}
where for any condition involving any or all of the test statistics given by $C(T^{(b)})$ (e.g., ($C(T^{(b)})=\{T_1^{(b)}>\tau\}$), we define
\begin{equation}\widehat{\qw_{H_0}}\{C(T^{(b)})\}=\frac{1}{B}\sum_{b=1}^B\mathbb{1}\{C(T^{(b)})\},\label{eq:empprob}\end{equation}
as the necessary empirical probability estimators. \citet{wt} demonstrated that the boundary from \citet{p77} is a special case of the more general boundaries they recommend. Following from their discussion, the popular boundaries due to \citet{obf} that use the empirically estimated boundary function $\widehat{\mathcal{B}}_{\mathrm{OBF}}(\tau)$ given by
\small
\begin{equation}\widehat{\qw_{H_0}}(T_1>\tau)+\sum_{m=2}^M\widehat{\qw_{H_0}}\left(T_m>\frac{\tau}{\sqrt{m}} \ \bigg| \ \bigcap_{c=1}^{m-1}\left\{T_c\leq\frac{\tau}{\sqrt{c}}\right\}\right)\widehat{\qw_{H_0}}\left(\bigcap_{c=1}^{m-1}\left\{T_c\leq\frac{\tau}{\sqrt{c}}\right\}\right),\label{eq:bfe}\end{equation}
\normalsize
in our case are a special case of their formulations. For the sake of brevity, we only consider the boundaries due to \citet{p77} and \citet{obf} herein since they are the most frequently used. Regardless of our choice of boundary function, we obtain an estimator for $\tau^*$ by solving 
\begin{equation}\hat{\tau}_*=\underset{\tau}{\mathrm{argmin}}\left\{|\widehat{\mathcal{B}}^*(\tau)-\alpha|\right\},\label{eq:best}\end{equation}
for either $\widehat{\mathcal{B}}^*(\tau)\in\{\widehat{\mathcal{B}}_{\mathrm{P}}(\tau),\widehat{\mathcal{B}}_{\mathrm{OBF}}(\tau)\}$, after which we set $\tau_m=\hat{\tau}_*$ for the \citet{p77} boundary and $\tau_m=\hat{\tau}_*/\sqrt{m}$ for the \citet{obf} boundary, respectively.

\subsection{Dynamic Boundary Estimation}
\label{ss34}

As we noted in the previous section, it is possible to use the procedure we developed using Theorem 1 in Subsection \ref{ss33} to estimate $\hat{\tau}_*$ through \eqref{eq:best} at any interim time. Traditional practices would suggest that one should simply compute these boundaries at the first interim time and use this boundary for the remaining analyses if we do not reject the null hypothesis at first. However, the methods introduced in Subsection \ref{ss33} allow us to dynamically re-compute the value of $\hat{\tau}_*$ (and hence all of the test statistic boundaries) at every interim time should we need to continue collecting data. To this end, let $\hat{\tau}_*^{(k)}$ be the estimator obtained through \eqref{eq:best} at time $t_k$ and let $\tau_m^{(k)}$ be the respective test statistic specific boundaries computed using this value. At interim $k$, it is worth noting that we can now use the updated $\tau_m^{(k)}$ values that make use of more available information as opposed to the $\tau_m^{(1)}$ values as traditional methodology would dictate. We compare our new dynamically updated approach to the traditional method in our simulations and application sections.

\subsection{Incorporating Partial Information}
\label{ss35}

If the repeated measures are longitudinal, at any interim analysis point, there will likely be individuals with partial information. One common source of missingness occurs by design, that is, insufficient time has elapsed between patients visits to warrant the collection of all desired observations. More generally, missingness may occur sporadically with total independence of the outcome, for which we say the data are missing completely at random (MCAR), or it may occur independently of the outcome conditional on the present covariates, for which we way the data are missing at random (MAR). \citet{lz} discusses why the traditional GEE process automatically handles data that are missing completely at random, but the other forms of missingness must be addressed. \citet{bvrt} utilizes inverse probability-weighted (IPW) methods to address missingness, but a more accessible approach is multiple imputation by chained equations (MICE) as implemented in the package MICE \citep{mice}. This methodology is founded upon the arguments from \citet{r87}. In short, this approach uses $L$ datasets, each of which has been completed using one of the $L$ imputations as specified. Coefficients for the analyst's model of choice are then computed alongside their variances, and a pooled estimate and variance are each computed. We can then substitute the pooled coefficients and variances into our test statistic expressions and proceed as discussed in Subsection \ref{ss33}, provided ample imputations (e.g., $L\geq30$) are used, since the $t$ distribution would well approximate a normal distribution for this number of imputations. If a significantly smaller number of imputations were used instead, we could not apply the delta method to an arbitrary function $h$ and only hypotheses concerning $\beta$ itself can be tested. In this case, following the arguments are \citet{r87}, the test statistics should be divided by $p$ and will consequently follow an an $\mathrm{F}(p,\nu_m)$ distribution, where $\nu_m$ is the degrees of freedom parameters defined in \citet{r87} for the $m$th analysis, which is computed automatically through MICE. To avoid this complication and preserve the ability to test more general hypotheses, we strongly suggest the use of a generous number of imputations whenever possible. Notably, however, the F test is not difficult to conduct when the hypothesis of interest concerns only a scalar parameter. Still, it should not be challenging to utilize a sufficient number of imputations in the presence of modest missingness (e.g., $10$ to $30$\% missing data), so the use of the F test is generally reserved for large datasets in which a scalar coefficient of interest is considered.

\section{Simulations}
\label{s4}

Our simulations are conducted in the R programming language version 4.4.2 \citep{R} and make notable use of the helpful package (\texttt{geex}) developed by \citet{sh}, which appropriates the theory of M-estimation as described and summarized by \citet{sb}. We also make use of the older package (\texttt{gee}) due to \citet{vjc}. The total sample sizes ($n$) vary based on each simulation setting, with values of either 400, 500, or 600. Each of these total groups yields precisely five observations (i.e., $s_i=5$ for all $i=1,\dots,n$). A covariate $Z_{ik}$ is drawn randomly from the $\mathrm{N}(1,1/16)$ distribution and is meant to simulate a low-variance nuisance covariate.  A binary variable $A_{ik}$ is randomly generated for each group for all times (i.e., $A_{i1}=\dots=A_{i5}\in\{0,1\}$). Each outcome $Y_{ik}$ is also presumed to be binary as discussed below.

Time is considered in two distinct ways, showcasing the flexibility of our methods, but in both cases, the actual elapsed time values are set to $1/12$, $3/12$, $6/12$, $1$, and $2$ for $\mathbb{T}_{i1}, \mathbb{T}_{i2}, \mathbb{T}_{i3}, \mathbb{T}_{i4},$ and $\mathbb{T}_{i5}$ for all $i=1,\dots, n$, respectively. The first model treats time as a continuous predictor and considers its interaction with treatment. This formulation is a generalization of the example provided in \citet{jtg} and elicits the model structure
\[\ew(Y_{ik}\mid X_{ik})=\mu(X_{ik},\beta)=\mathrm{expit}(\beta_0+\beta_AA_{ik}+\beta_T\mathbb{T}_{ik}+\beta_{AT}A_{ik}\mathbb{T}_{ik}+\beta_ZZ_{ik}),\]
where $X_{ik}=(A_{ik}, \mathbb{T}_{ik}, Z_{ik})^\T$, $\beta=(\beta_0, \beta_A, \beta_T, \beta_{AT}, \beta_Z)^\T$ and $\mathrm{expit}(x)=(1+e^{-x})^{-1}$. Within group covariances are computed in a manner similar to the simulations from \citet{jtg}, with a few modifications. Letting $\mathrm{logit}(x)=\mathrm{expit}^{-1}(x)$, we let $\sigma_{ikr}=\sigma_{irk}=\mathrm{Cov}\left\{\mathrm{logit}(Y_{ik}),\mathrm{logit}(Y_{ir})\right\}=e^{-|\mathbb{T}_{ik}-\mathbb{T}_{ir}|}.$ To this end, let $W_i=(\sigma_{ikr})_{k,r=1,\dots,5}$ be the covariance matrix for group $i$ as generated. We note that if the times were equidistant, this would be an AR-1 structure, which is commonly used as a working correlation structure in clinical settings. We consider independent and exchangeable structures, both of which are obviously misspecified, even under equidistant times. This facilitates an examination of the robustness of our method in the face of poor modeling choices.

Denoting $\lambda(X_{ik},\beta)=\mathrm{logit}\{\mu(X_{ik},\beta)\}$, and $\lambda_i=(\lambda(X_{i1},\beta),\dots,\lambda(X_{i5},\beta))^\T$, the outcomes are generated using $(\mathrm{logit}(Y_{i1}),\dots,\mathrm{logit}(Y_{i5}))^\T\sim\mathrm{MVN}_5(\lambda_i, W_i)$ for all $i=1,\dots,n$. We then set $\pi_{ik}=\mathrm{expit}\{\mathrm{logit}(Y_{ik})\}$ and draw $Y_{ik}\sim\mathrm{Ber}(\pi_{ik})$ for every $k=1,\dots,5$ and $i=1,\dots,n$. Accordingly, $\{\{(X_{ik},Y_{ik})\}_{k=1}^5\}_{i=1}^n$ is our simulated dataset.

The second time model bears a closer resemblance to the simulations conducted by \citet{jtg}. Times are treated as discrete, categorized, ordinal variables reflected by the usage of indicator random variables. The second model is accordingly 
\[\mu(X_{ik},\beta)=\mathrm{expit}\left[\beta_0+\beta_AA_{ik}+\beta_ZZ_{ik}+\sum_{t=2}^5\left\{\beta_t\mathbb{1}(k=t)+\beta_{At}A_{ik}\mathbb{1}(k=t)\right\}\right]\]
where the covariate vector is now adjusted such that $X_{ik}=(A_{ik}, Z_{ik}, \{\mathbb{1}(t=k)\}_{k=2}^5)^\T$ and similarly $\beta=(\beta_0,\beta_A,\beta_Z,\{\beta_t\}_{t=2}^5,\{\beta_{At}\}_{t=2}^5)^\T.$ Here, $k=1$ is used as an implicit baseline, which concords with the fact that there is neither a $\beta_1$ nor $\beta_{A1}$ coefficient in our model. We note that the gaps between times are not reflected in this model, but they are still used to construct the within-group covariance matrix as before. The data generation process for the treatment indicator, nuisance covariate, correlation structure, and outcome generation is identical. We refer to each of these models as the continuous time model and discrete time model respectively.

In the continuous time model, we fix $\beta_0=\beta_A=\beta_X=0.1$ and we set $\beta_t=-0.1$. The value of our interaction coefficient $\beta_{AT}$ is zero in our type I error simulations, though it changes for our different power level runs. In the discrete time model, we fix $\beta_0=\beta_A=\beta_X=0.1$ and set $\beta_t=-0.1$ for all $t=2,\dots,5$. Similarly, the values of $\beta_{At}$ for $t=2,\dots,5$ are zero in our type I error simulations and change our power level runs. For each time model, we test whether or not a treatment-time interaction effect exists. In the continuous case, this simply amounts to testing $H_0: \beta_{AT}=0$ against $H_0: \beta_{AT}\neq0$. In the discrete case, the hypothesis is slightly more complicated, taking the form $H_0: Q\beta=\mathbb{0}_4$ against $H_a: Q\beta\neq\mathbb{0}_4$, where $Q$ is a $4\times11$ matrix given in block form as $Q=[\boldsymbol{0}_{4\times7} \ \boldsymbol {I}_4]$, wherein $\boldsymbol{0}_{4\times7}$ is a $4\times 7$ matrix of zeroes and $\boldsymbol {I}_4$ is a $4\times4$ identity matrix. This is merely a way of testing whether or not any of the coefficients $\{\beta_{At}\}_{t=2}^5$ are significantly different than zero.

Finally, a third scenario is considered to examine our method's performance under standard missingness structures. The model and data generation procedure are exactly the same as for the standard continuous model above, but we introduce additional complexities to simulate the natural incompleteness of real datasets. Both missingness by design and missing at random are considered as processes for our case. Missingness by design is facilitated through the creation of an elapsed time variable $\mathcal{T}_{ik}$ by using the $\mathbb{T}_{ik}$ variables. Specifically, we set $\mathcal{T}_{1k}=\mathbb{T}_{1k}$ for $k=1,\dots, 5$ Then, we recursively define $\mathcal{T}_{i1}=\mathcal{T}_{(i-1)1}+E$, $\mathcal{T}_{i2}=\mathcal{T}_{i1}+2/12$, $\mathcal{T}_{i3}=\mathcal{T}_{i1}+5/12$, $\mathcal{T}_{i4}=\mathcal{T}_{i1}+11/12$, and $\mathcal{T}_{i5}=\mathcal{T}_{i1}+23/12$ for all $i\geq2$, where $E\sim\mathrm{Exp}(96)$ is a randomly drawn exponential random variable. After 85 and 185 patients have provided their baseline observations, 1.25 time units are added to the elapsed time to simulate a pause in recruitment and allow for full sets of observations to be collected before the first and second interims, respectively. This simulates the ordered arrival of groups to be considered and attempts to mirror a clinical setting. Once the first observation from the $n_m$th group has been recorded at time $\mathcal{T}_{n_m1}$, all rows of the dataset corresponding to an elapsed time $\mathcal{T}_{ik}$ such that $\mathcal{T}_{ik}>\mathcal{T}_{n_m1}$ are rendered incomplete, such that their corresponding $Y_{ik}$ and $Z_{ik}$ values are missing. This demonstrates the process of missingness by design, whereby an insufficient amount of time has elapsed for all subjects to provide all five necessary observations. Simultaneously, another form of missingness is also introduced into the data. For every $Z_{ik}$, we draw the Bernoulli variables $R_{ik}\sim\mathrm{Ber}\{\mathrm{expit}(2.1-Z_{ik})\}$ for ``high'' missingness and $R_{ik}\sim\mathrm{Ber}\{\mathrm{expit}(3.1-Z_{ik})\}$ for ``low'' missingness. If $R_{ik}=0$, then $Y_{ik}$ and $Z_{ik}$ are missing, and if $R_{ik}=1$, they are left unobstructed. This is a simple way of injecting missingness at random (MAR) into our dataset. To handle these two forms of missingness, we use the MICE package \citet{mice} with 30 imputations per analysis with random forest imputation as provided through the package's technical argumentation. GEE coefficient estimates and their standard errors are then pooled according to the recommendations of \citet{r87} and these estimates and standard errors are then used in our method as described in Section \ref{s3}. Imputations that yield singular covariances are discarded and replaced by nonpathological imputations.

In all simulations, 1000 Monte Carlo iterations were used to compute an empirical estimate of the Pocock and OBF boundaries as well as to approximate power or type I error.

\begin{table}
\caption{Empirical type I error estimates for the interaction test using static and dynamic Pocock and O'Brien-Fleming (OBF) boundaries under various effect sizes, sample sizes, and working correlation structures for the continuous and discrete time models with no missing data. Ind. = Independent; Exc. = Exchangeable.}
\label{t1}
\begin{center}
\begin{tabular}{l|cccccccc}
\hline\hline
 &  &  & & \multicolumn{2}{c}{\underline{Pocock}}& \multicolumn{2}{c}{\underline{OBF}}\\
Model & \multicolumn{1}{c}{$n$} & \multicolumn{1}{c}{$V_i$} & \multicolumn{1}{c}{Na\"ive} & \multicolumn{1}{c}{Static} &  \multicolumn{1}{c}{Dynamic} & \multicolumn{1}{c}{Static} &  \multicolumn{1}{c}{Dynamic}\\ \hline
 Continuous & 400&Ind.&0.108&0.050&0.050&0.059&0.054\\
& &Exc.&0.102&0.050&0.051&0.055&0.053\\\hline
 Continuous & 500&Ind.&0.112&0.053&0.057&0.056&0.052 \\
& &Exc.&0.114&0.055&0.054&0.060&0.058\\ \hline
 Continuous & 600&Ind.&0.100&0.051&0.049&0.048&0.045\\
& &Exc.&0.114&0.052&0.052&0.049&0.050\\\hline\hline
Discrete & 400&Ind.&0.115&0.057&0.054&0.055&0.056\\
 & &Exc.&0.109&0.054&0.052&0.049&0.051\\ \hline
Discrete & 500 &Ind.&0.113&0.059&0.058&0.051&0.048\\
& &Exc.&0.128&0.068&0.069&0.062&0.069\\ \hline
Discrete & 600&Ind.&0.128&0.055&0.055&0.063&0.063\\
 & &Exc.&0.101&0.049&0.049&0.047&0.047\\
\hline
\end{tabular}
\end{center}
\end{table}

\begin{table}
\caption{Empirical type I error estimates for the interaction test using static and dynamic Pocock and O'Brien-Fleming (OBF) boundaries under various effect sizes, sample sizes, and working correlation structures for the continuous time model involving missing data. Ind. = Independent; Exc. = Exchangeable.}
\label{t2}
\begin{center}
\begin{tabular}{l|cccccccc}
\hline\hline
 &  &  & & \multicolumn{2}{c}{\underline{Pocock}}& \multicolumn{2}{c}{\underline{OBF}}\\
Missingness & \multicolumn{1}{c}{$n$} & \multicolumn{1}{c}{$V_i$} & \multicolumn{1}{c}{Na\"ive} & \multicolumn{1}{c}{Static} &  \multicolumn{1}{c}{Dynamic} & \multicolumn{1}{c}{Static} &  \multicolumn{1}{c}{Dynamic}\\ \hline
Low  & 400 &Ind.&0.117&0.060&0.059&0.052&0.052\\
& &Exc.&0.152&0.065&0.067&0.069&0.070\\\hline
Low  & 500 &Ind.&0.126&0.066&0.068&0.071&0.067\\
& &Exc.&0.133&0.064&0.065&0.059&0.057\\\hline
Low  & 600 &Ind.&0.103&0.055&0.058&0.061&0.058\\
& &Exc.&0.126&0.066&0.066&0.066&0.063\\\hline\hline
High  & 400 &Ind.&0.116&0.053&0.056&0.056&0.058\\
& &Exc.&0.140&0.068&0.067&0.075&0.073\\\hline
High  & 500 &Ind.&0.120&0.060&0.058&0.056&0.056\\
& &Exc.&0.132&0.075&0.079&0.071&0.077\\\hline
High  & 600 &Ind.&0.113&0.056&0.055&0.055&0.059\\
& &Exc.&0.141&0.071&0.068&0.064&0.066\\\hline
\end{tabular}
\end{center}
\end{table}

\begin{table}
\caption{Empirical power estimates for the interaction test using static and dynamic Pocock and O'Brien-Fleming (OBF) boundaries under various effect sizes, sample sizes, and working correlation structures for the continuous time model. Ind. = Independent; Exc. = Exchangeable.}
\label{t3}
\begin{center}
\begin{tabular}{cccccccc}
\hline\hline
& & & \multicolumn{2}{c}{\underline{Pocock}} & \multicolumn{2}{c}{\underline{OBF}}\\
$n$ & \multicolumn{1}{c}{$\beta_{AT}$} & \multicolumn{1}{c}{$V_i$} & \multicolumn{1}{c}{Static} &  \multicolumn{1}{c}{Dynamic} & \multicolumn{1}{c}{Static} &  \multicolumn{1}{c}{Dynamic}\\ \hline
400&$-0.40$&Ind.&0.613&0.614&0.671&0.667\\
&&Exc.&0.619&0.619&0.675&0.677\\\hline
400&$-0.45$&Ind.&0.731&0.731&0.770&0.771\\
&&Exc.&0.724&0.727&0.766&0.770\\\hline
400&$-0.5$&Ind.&0.822&0.817&0.865&0.863\\
&&Exc.&0.821&0.820&0.852&0.851\\\hline\hline
500&$-0.40$&Ind.&0.729&0.724&0.768&0.772\\
&&Exc.&0.737&0.740&0.790&0.791\\\hline
500&$-0.45$&Ind.&0.832&0.828&0.851&0.855\\
&&Exc.&0.838&0.841&0.876&0.873\\\hline
500&$-0.50$&Ind.&0.886&0.885&0.918&0.917\\
&&Exc.&0.893&0.896&0.920&0.928\\\hline\hline
600&$-0.40$&Ind.&0.798&0.803&0.836&0.837\\
&&Exc.&0.794&0.795&0.834&0.834\\\hline
600&$-0.45$&Ind.&0.887&0.888&0.923&0.921\\
&&Exc.&0.886&0.891&0.921&0.920\\\hline
600&$-0.50$&Ind.&0.936&0.943&0.962&0.958\\
&&Exc.&0.946&0.947&0.960&0.963\\\hline\hline
\end{tabular}
\end{center}
\end{table}

\begin{table}
\caption{Empirical power estimates for the interaction test using static and dynamic Pocock and O'Brien-Fleming (OBF) boundaries under various effect sizes, sample sizes, and working correlation structures for the discrete time model. Ind. = Independent; Exc. = Exchangeable; $v=(3/12,6/12,1,2)^\T$.}
\label{t4}
\begin{center}
\begin{tabular}{ccccccc}
\hline\hline
&&&\multicolumn{2}{c}{\underline{Pocock}}
&\multicolumn{2}{c}{\underline{OBF}}\\
$n$ & \multicolumn{1}{c}{$\beta_{At}$} & \multicolumn{1}{c}{$V_i$} & \multicolumn{1}{c}{Static} &  \multicolumn{1}{c}{Dynamic}& \multicolumn{1}{c}{Static} &  \multicolumn{1}{c}{Dynamic}\\ \hline
400&$-0.40v$&Ind. &0.428&0.441&0.503&0.502\\
&&Exc. &0.453&0.459&0.492&0.486\\\hline
400&$-0.45v$&Ind. &0.529&0.527&0.593&0.594\\
&&Exc. &0.531&0.527&0.597&0.608\\\hline
400&$-0.50v$&Ind. &0.653&0.658&0.729&0.727\\
&&Exc. &0.655&0.650&0.710&0.705\\\hline\hline
500&$-0.40v$&Ind. &0.561&0.559&0.621&0.620\\
&&Exc. &0.541&0.539&0.611&0.612\\\hline
500&$-0.45v$&Ind. &0.639&0.633&0.710&0.708\\
&&Exc. &0.671&0.676&0.736&0.726\\\hline
500&$-0.50v$&Ind. &0.751&0.755&0.809&0.812\\
&&Exc. &0.780&0.781&0.838&0.831\\\hline\hline
600&$-0.40v$&Ind. &0.622&0.623&0.690&0.679\\
&&Exc. &0.620&0.626&0.685&0.689\\\hline
600&$-0.45v$&Ind. &0.757&0.752&0.822&0.824\\
&&Exc. &0.746&0.747&0.806&0.801\\\hline
600&$-0.50v$&Ind. &0.846&0.851&0.893&0.896\\
&&Exc. &0.859&0.863&0.894&0.897\\\hline\hline
\end{tabular}
\end{center}
\end{table}

\begin{table}
\caption{Empirical power estimates for the interaction test using static and dynamic Pocock and O'Brien-Fleming (OBF) boundaries under various effect sizes, sample sizes, and working correlation structures for the continuous time model involving missing data. Ind. = Independent; Exc. = Exchangeable.}
\label{t5}
\begin{center}
\begin{tabular}{l|cccccccc}
\hline\hline
& & & & \multicolumn{2}{c}{\underline{Pocock}} & \multicolumn{2}{c}{\underline{OBF}}\\
Missingness& $n$ & \multicolumn{1}{c}{$\beta_{AT}$} & \multicolumn{1}{c}{$V_i$} & \multicolumn{1}{c}{Static} &  \multicolumn{1}{c}{Dynamic} & \multicolumn{1}{c}{Static} &  \multicolumn{1}{c}{Dynamic}\\ \hline
Low&400&$-0.40$&Ind.&0.592&0.594&0.640&0.644\\
&&&Exc.&0.586&0.584&0.636&0.634\\\hline
Low&400&$-0.45$&Ind.&0.666&0.662&0.718&0.713\\
&&&Exc.&0.700&0.702&0.759&0.764\\\hline
Low&400&$-0.5$&Ind.&0.764&0.772&0.828&0.823\\
&&&Exc.&0.789&0.786&0.825&0.824\\\hline\hline
Low&500&$-0.40$&Ind.&0.689&0.689&0.735&0.734\\
&&&Exc.&0.698&0.705&0.743&0.745\\\hline
Low&500&$-0.45$&Ind.&0.777&0.787&0.836&0.837\\
&&&Exc.&0.805&0.803&0.844&0.848\\\hline
Low&500&$-0.50$&Ind.&0.851&0.857&0.887&0.892\\
&&&Exc.&0.871&0.872&0.914&0.912\\\hline\hline
Low&600&$-0.40$&Ind.&0.778&0.768&0.814&0.820\\
&&&Exc.&0.801&0.803&0.823&0.823\\\hline
Low&600&$-0.45$&Ind.&0.846&0.845&0.880&0.881\\
&&&Exc.&0.866&0.871&0.904&0.902\\\hline
Low&600&$-0.50$&Ind.&0.917&0.911&0.932&0.934\\
&&&Exc.&0.941&0.949&0.963&0.962\\\hline\hline
High&400&$-0.40$&Ind.&0.483&0.488&0.533&0.531\\
&&&Exc.&0.544&0.541&0.604&0.598\\\hline
High&400&$-0.45$&Ind.&0.626&0.624&0.668&0.668\\
&&&Exc.&0.622&0.625&0.672&0.664\\\hline
High&400&$-0.5$&Ind.&0.687&0.698&0.747&0.751\\
&&&Exc.&0.709&0.714&0.757&0.762\\\hline\hline
High&500&$-0.40$&Ind.&0.583&0.584&0.648&0.643\\
&&&Exc.&0.627&0.631&0.675&0.674\\\hline
High&500&$-0.45$&Ind.&0.707&0.711&0.766&0.765\\
&&&Exc.&0.733&0.736&0.782&0.780\\\hline
High&500&$-0.50$&Ind.&0.784&0.785&0.828&0.835\\
&&&Exc.&0.843&0.844&0.875&0.880\\\hline\hline
High&600&$-0.40$&Ind.&0.683&0.676&0.746&0.747\\
&&&Exc.&0.708&0.708&0.763&0.754\\\hline
High&600&$-0.45$&Ind.&0.794&0.791&0.843&0.840\\
&&&Exc.&0.813&0.820&0.857&0.857\\\hline
High&600&$-0.50$&Ind.&0.856&0.854&0.893&0.886\\
&&&Exc.&0.868&0.865&0.899&0.898\\\hline\hline
\end{tabular}
\end{center}
\end{table}

Tables \ref{t1} and \ref{t2} present the type I error rates for the continuous and discrete complete data scenarios, as well as the low and high missingness scenarios, for a test with nominal $5\%$ type I error. In all cases, the na\"ive test statistics exhibit highly inflated type I error, underscoring the necessity of sequential methods to control type I error. Across both tables, the differences in type I error rates between the Pocock and O’Brien–Fleming (OBF) boundaries are negligible, as are the differences between static and dynamic boundary calculations.
Table \ref{t2} shows somewhat less precise type I error estimates because of the additional variability incurred from the use of multiple imputations. Nonetheless, proper coverage is still attained asymptotically. The empirical type I error rates shown in Table \ref{t1} are highly consistent across sample sizes, working correlation structure, and boundary type, ranging from 0.047 to 0.069, with the variation attributable to Monte Carlo noise. Corresponding estimates in Table \ref{t2} are slightly more variable, reaching a maximum of 0.079 in the most extreme case. This increased variability due to noise reflects the additional uncertainty associated with the high-missingness scenario. Nevertheless, the type I error of the test is well-maintained.

Tables \ref{t3} and \ref{t4} present the estimated power for the complete continuous and discrete models, respectively, under a range of alternative scenarios. In general, the choice of static versus dynamic boundary computation, or the choice of working correlation structure has negligible impact on the power of any given test. Although the OBF boundaries show a slight power advantage at smaller sample sizes (especially for large effects), this difference diminishes as the sample size increases. The power of the test increases with the effect size as expected. The discrete model consistently yields lower power than the continuous model, reflecting its greater complexity and formulation. For the lowest interaction effect size $\beta_{AT} = -0.40$  in Table \ref{t3} , using static boundaries from Pocock method, we achieve an empirical power of 0.613 for $n=400$, which increases to 0.798 for $n=600$. For the highest effect size of $\beta_{AT} = -0.50$, the dynamic OBF boundaries resulted in  empirical power estimates of 0.851 with $n=400$, which increases to 0.963 with $n$ increasing to $600$, both with an exchangeable working correlation structure.. Table \ref{t4} yields similar results, with a power of 0.423 at the lowest and 0.897 at the highest, reflecting a greater necessary burden to ascend to higher power in the case of the most complex, discrete-time model. 

Finally, Table \ref{t5} showcases the power in the missing data scenarios. A small loss of power due to missing data imputation is exhibited, compared to the full-data scenario presented in Table \ref{t3}. Table \ref{t5} demonstrates nearly identical power ranges to those in Table \ref{t3} for the low missingness model, with the same upper bound for empirical power, while the high missingness model rises to an empirical power of 0.899 at the highest effect size with static OBF boundary at $600$ sample size, demonstrating the small impact of missing data on power.

\section{Application to the VIRAHEP-C Study}
\label{s5}

We demonstrate the utility of our method by using it to analyze the VIRAHEP-C dataset as provided by the National Institute for Diabetes and Digestive and Kidney Diseases at the National Institute of Health. The study aims to investigate the potential effect of race on the efficacy of antiviral therapy for hepatitis C. The specific treatment in question is a weekly 180$\mu$g dose of pegylated interferon alfa-2a alongside a daily 1000-1200mg dose of ribavirin. Greater details are provided in \citet{hepc} and\citet{hepc2}. A total of 205 Caucasian-American and 196 African-American treatment-naive hepatitis C genotype I patients were recruited who received this treatment starting at baseline. The viral load count was the primary measurement used to assess the efficacy of the treatment, with achieving an undetectable viral load being the treatment target. Thus, our outcome of interest is the binary variable $Y$, which is an indicator of whether or not a patient currently exhibits a detectable ITT viral load count or not. Besides race being the main covariate of interest, we aimed to control for age, sex, and body mass index (BMI) in the analysis. 

Early efficacy of antiviral therapy for Hepatitis C is assessed over the first 28 days of therapy \citep{hepc}. In this study over the span of 28 days, most subjects have measurements on days 0, 1, 2, 7, 14, and 28. However, each subject does not necessarily have the same number of observations as the others, yielding a heterogeneity of group sizes . Any missing covariates are presumed to be missing at random. Though we have the full dataset in our possession, realistic studies would have to plan interim analyses in advance. Hence, we emulate this process by conducting our first interim analysis after the 134th patient has provided their first observation and our second interim analysis after the 269th patient has provided their first observation. Observations whose corresponding collection time surpasses the interim analysis time were assumed to have outcomes ($Y$) missing by design.

We select an exchangeable within-subject variance structure and set $\ew(Y_{ik}\mid X_{ik})$ as
\small
\[\mathrm{expit}\left(\beta_0+\beta_R\mathrm{Race}_{i}+\beta_S\mathrm{Sex}_{i}+\beta_E\mathrm{Year}_i+\beta_T\mathrm{Time}_{ik}+\beta_B\mathrm{BMI}_{i}+\beta_I\mathrm{Race}_{i}\mathrm{Time}_{ik}\right),\]
\normalsize
where $i$ indexes individuals, $k$ indexes time of observation for an individual, $\mathrm{Race}_{i}$ is an indicator with value 1 if the $i$th subject is African-American, $\mathrm{Year}_i$ is the $i$th subject's birth year, $\mathrm{Sex}_{i}$ indicates if the $i$th subject is female, and $X_{ik}=(\mathrm{Race}_{i}, \mathrm{Sex}_{i}, \mathrm{Time}_{ik}, \mathrm{BMI}_{i},\mathrm{Year}_i)^\T$ is the aggregate covariate vector for the $i$th individual at a specific collection time $k$. The analysis was again conducted R, this time in version 4.5.2 due to updates. Figure \ref{f1} shows the empirical proportions of subjects with detectable viral load counts over time by race at each analysis time. We note that at the first two analysis times, the data show a widening of the gap between the two lines over time suggesting a potentially significant difference in the rate of decline in proportion with detectable viral load between African-Americans and Caucasian-Americans, while the third does not show this phenomenon. A reasonable research question is whether this racial difference observed is statistically significant, which is equivalent to testing $H_0:\beta_I=0$ against $H_a:\beta_I\neq0$.
\begin{figure}[!ht]
    \centering
    \includegraphics[width=12cm]{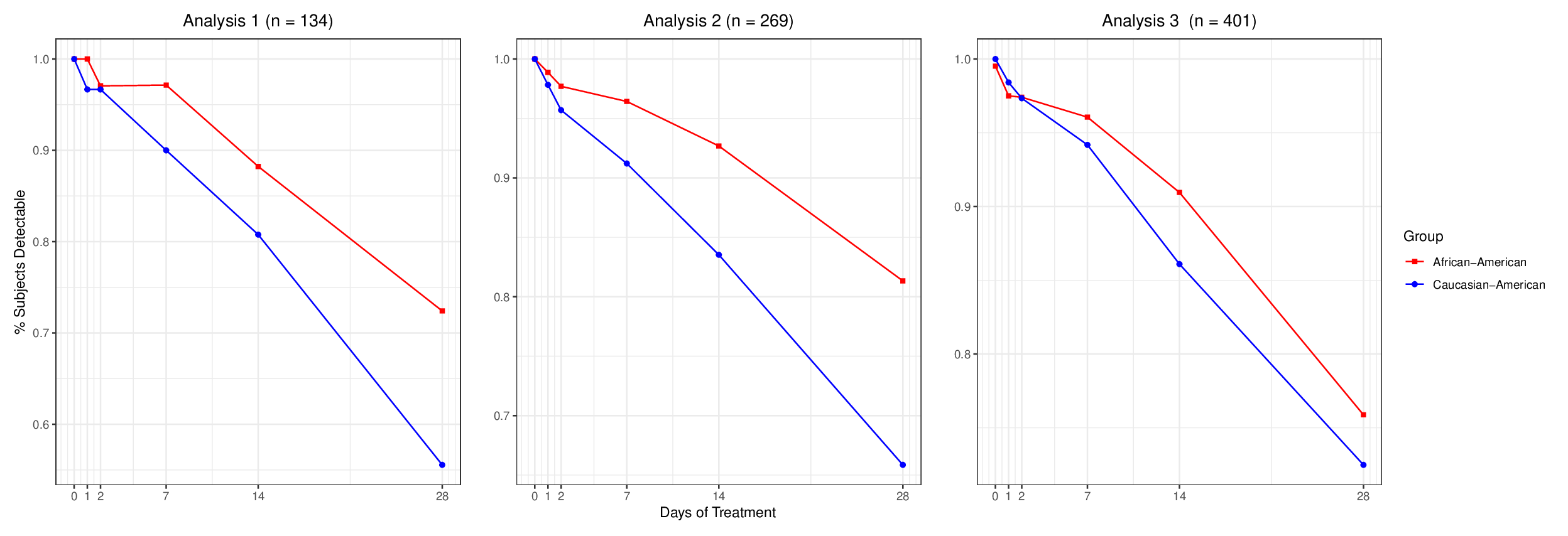}
    \caption{Empirical detectability proportions for African-Americans and Caucasian-Americans at Each Analysis Time.}
    \label{f1}
\end{figure}
We used both the Pocock and OBF boundaries in our hypothetical interim analyses to see if there would be a discernible difference between implementing one or the other. A seed was fixed and one million Monte Carlo iterations were used to compute the boundary, with thirty imputations for missing data. The static Pocock and OBF boundaries we computed are $5.235$ and $(7.524, 5.320, 4.344)$, respectively. We obtained a test statistic of $0.003$ and continued past the first interim. The dynamic Pocock and OBF boundaries from the second analysis were $5.233$ and $5.302$, respectively, and we obtained a test statistic of $0.098$. We again fail to reject $H_0$ by all boundary computations and continue. At the last analysis, we obtain a dynamic Pocock and OBF boundaries of $5.227$ and $4.329$ with a test statistic of $0.046$. Hence, across all three analysis, we fail to reject the null hypothesis, which suggests that there is no statistically significant interaction effect between race and time on the early efficacy of hepatitis C antiviral therapy in terms of detectability. We note that this conclusion was heavily influenced by the apparent noisiness of the data (e.g., the second interim estimate for $\beta_I$ was $-0.042$ with a standard error of $0.133$, yielding the largest test statistic).

\section{Discussion}
\label{s6}
We have constructed a general theoretical framework for robust sequential analysis in settings involving longitudinal and clustered outcomes that also accommodates missingness at random and missingness by design. Our testing procedure exhibits proper coverage and desirable power even under misspecification and allows for the intuitive generalization of multiple imputation with chained equations to the setting of interest. Moreover, the computational effort required to carry out these procedures is non-intensive, as the technical framework for obtaining standard robust covariance matrices from a GEE analysis following \citet{lz} is already in place and widely available. More importantly, our approach allows for any class of hypotheses as specified in Sections \ref{s2} and \ref{s3} to be tested and does not demand the imposition of a nuisance parameter-treatment parameter dichotomy when constructing the model of interest. This development permits a greater level of freedom when conducting such analyses, especially when compared to the other existing approaches like those developed by \citet{wsl}, \citet{lkt}, and more recently by \citet{jtg}. Future work could focus on extending our methodology to incorporate $\alpha$ spending functions \citep{ld}, another popular way of progressively partitioning type I error as interim analyses proceed, as opposed to fixing the available information at each analysis \textit{a priori}.

\section*{Acknowledgements}

This work was partially funded through a Patient-Centered Outcomes Research Institute (PCORI) Award (ME-2021C3-24215). The statements in this work are solely the responsibility of the authors and do not necessarily represent the views of PCORI, its Board of Governors or the Methodology Committee. We thank all the stakeholder advisory board members of the study: Jordan Karp of the University of Arizona, Julie Bauman of the George Washington University, Peter Thall of the MD Anderson Cancer Center, Douglas Landsittel of the University of Buffalo, and Gong Tang, Joyce Chang, and Meredith Lotz Wallace of the University of Pittsburgh, for their input on the project.

\vspace*{-8pt} 

\bibliographystyle{biom} 
\bibliography{References}

\label{lastpage}

\end{document}